\documentclass[letterpaper,11pt]{article}
\usepackage[margin=1in]{geometry}
\usepackage[utf8]{inputenc} % allow utf-8 input
\usepackage[T1]{fontenc}    % use 8-bit T1 fonts
\usepackage{libertine}    % Libertine font
\usepackage{libertinust1math}   % Libertine math font

\usepackage{macros}

\title{\LARGE \bfseries Hypothesis Selection via Weighted Medians}

\author{
    \begin{minipage}[t]{0.45\textwidth}
        \centering
        {\normalsize Liang-Yu Zou} \\
        {\small \textit{Aarhus University}} \\
        {\small$\mathtt{zou@cs.au.dk}$}
    \end{minipage}
}

\date{\vspace*{-1em}}

\begin{document}

% \hypersetup{pageanchor=false}
% \pagenumbering{gobble}          % Suppress page numbers (no page number displayed)

\maketitle
              
\begin{abstract}
Given $n$ distributions (hypotheses) $\mathcal H=\{H_1,\ldots,H_n\}$ and samples from an unknown target distribution $P$, proper hypothesis selection seeks a distribution whose total variation distance to $P$ is close to the best achievable within the known distributions.
We give an algorithm that, using $O(\varepsilon^{-2}\log(n/\delta))$ target samples, running in $O(n\varepsilon^{-2}\log(n/\delta))$ time, with probability at least $1-\delta$, returns $H_{\hat i}$ satisfying $d_\mathrm{TV}(P,H_{\hat i})\leq 3\min_{i\in[n]}d_\mathrm{TV}(P,H_i)+\varepsilon.$
The algorithm achieves the optimal approximation factor and sample complexity.
The running time improves the dependence on the failure probability from polynomial to logarithmic in the state-of-the-art bound $\tilde{O}(n\varepsilon^{-2}\delta^{-1})$ of Aamand, Aliakbarpour, Chen, and Silwal (COLT 2026), thereby answering their Question~2.
\end{abstract}
  
% \newpage                      

% \tableofcontents

% \newpage                       

% \pagenumbering{arabic}          
% \setcounter{page}{1}            
% \hypersetup{pageanchor=true}

\section{Introduction}
\label{sec:introduction}

Selecting a probability distribution from a finite class is a basic task in statistical learning.
The target distribution need not belong to the class: the aim is to compete with the best available distribution using samples from the target.
In hypothesis selection, the quality of this approximation is measured in total variation distance, and the challenge is to obtain statistical guarantees with a running time that scales favorably with the number of distributions.

Let $\cH=\{H_1,\ldots,H_n\}$ be a class of \emph{distributions}, each a probability measure on a measurable space $(\cX,\cB)$, and let $P$ be the unknown target distribution.
Fix a common dominating measure $\mu$ for $P,H_1,\ldots,H_n$, and write $h_i=dH_i/d\mu$ for the density of $H_i$.
We have sample access to $P$ and sample and density-query access to every $H_i$; the density of $P$ is not available to the algorithm.
Writing $\OPT=\min_{i\in[n]}\TV(P,H_i)$, the goal is to return a probability measure $Q$ such that, with probability at least $1-\delta$,
\begin{equation*}
    \TV(P,Q)\leq C\cdot\OPT+\eps.
\end{equation*}
Here $\eps\in (0,1)$ is the additive error and $\delta\in(0,1)$ is the failure probability.
The selection is \emph{proper} if $Q\in\cH$; otherwise, the algorithm is allowed to be \emph{improper}.

\paragraph{Statistical guarantees.}
Classical minimum-distance methods achieve approximation factor $C=3$ using $O(\eps^{-2}\log(n/\delta))$ samples from $P$~\cite{yatracos1985rates,devroye2001combinatorial}.
The factor $3$ is also optimal for proper selection with a sample size independent of the domain size~\cite{bousquet2019optimal,aamand2025fast}.
Allowing an improper output improves the achievable factor to $2$~\cite{bousquet2019optimal,bousquet2022statistically}.
Thus, for proper selection, the remaining issue is how efficiently the optimal statistical guarantee can be attained.

\paragraph{Computational guarantees.}
Classical factor-$3$ selectors have cubic or quadratic running times in $n$~\cite{devroye2001combinatorial,mahalanabis2007density}.
Aliakbarpour, Bun, and Smith~\cite{aliakbarpour2024optimal} obtained a nearly linear dependence on $n$, with running time $\tilde O(n/(\delta^3\eps^3))$ and optimal target-sample complexity.
Aamand, Aliakbarpour, Chen, and Silwal~\cite{aamand2025fast} improved the running time to $\tilde O(n/(\delta\eps^2))$ for $\delta>1/n$.
Selected bounds are compared in \cref{tab:prior}.
Aamand et al.\ ask whether proper selection can simultaneously achieve factor $3$, $O(\eps^{-2}\log(n/\delta))$ target samples, and running time $\tilde O(n\log(1/\delta)/\eps^2)$~\cite[Question~2]{aamand2025fast}.
We answer this question affirmatively.

\begin{table}[t]
\centering
\small
\renewcommand{\arraystretch}{1.2}
\begin{tabular}{@{}llll@{}}
\toprule
Algorithm & Running time & Failure probability& Comments\\
\midrule
Minimum distance~\cite{yatracos1985rates,devroye2001combinatorial}
    & $O(n^3\eps^{-2}\log n)$ & $1/\mathrm{poly}(n)$& \\
Modified minimum distance~\cite{mahalanabis2007density}
    & $O(n^2\eps^{-2}\log n)$ & $1/\mathrm{poly}(n)$& \\
Approx min distance~\cite{aliakbarpour2024optimal}
    & $\tilde O(n/(\delta^3\eps^3))$ & $\delta$& \\
Semi-distance threshold~\cite[Theorem~1.3]{aamand2025fast}
    & $\tilde O(n/(\delta\eps^2))$ & $\delta$& $\delta>1/n$\\
\midrule
This paper, \cref{thm:proper}
    & $O(n\eps^{-2}\log(n/\delta))$ & $\delta$&\\
\bottomrule
\end{tabular}
\caption{Proper hypothesis selection with approximation factor $3$.
All rows use $O(\eps^{-2}\log(n/\delta))$ target samples; for polynomially small failure probability, this is $O(\eps^{-2}\log n)$.
More results are collected in~\cite[Table~1]{aamand2025fast}.}
\label{tab:prior}
\end{table}

\begin{theorem}[Fast proper hypothesis selection]
\label{thm:proper}
Let $n\geq1$, $0<\eps\leq1$, and $0<\delta<1$. 
With probability at least $1-\delta$, \cref{alg:proper} returns a distribution $H_{\hat i}\in\cH$ which satisfies 
\[
    \TV(P,H_{\hat i})\leq3\OPT+\eps.
\]
The algorithm uses $O(\eps^{-2}\log(n/\delta))$ target samples, and runs in $O(n\eps^{-2}\log(n/\delta))$ time.
\end{theorem}

\subsection{Overview of the method}
\label{sec:overview}
The algorithm compares the input distributions with the target through a common median measure.
Starting from uniform weights, it updates the weights so that the median measures are close, on average, to a best input distribution.

At round $t$, the algorithm maintains positive weights $\bw_t$ summing to one.
For each point $x$, let $u_t(x)$ be the lower weighted median of $h_1(x),\ldots,h_n(x)$ under these weights, and let $U_t$ be the measure with density $u_t$.
Since $U_t$ need not have mass one, we measure its distance to $H_i$ by $D_{t,i}=\frac12\|U_t-H_i\|_1$, where $\|\cdot\|_1$ denotes the total variation norm.

Balanced median tests $\sigma_{t,i}(x)$ indicate whether $h_i(x)$ lies below or above $u_t(x)$, with fractional tie values ensuring $\sum_iw_{t,i}\sigma_{t,i}(x)=0$.
These tests yield gains $g_{t,i}$ that guide multiplicative weight updates: positive gains increase the corresponding weights, while negative gains decrease them.
Since the median measure is not normalized and is difficult to sample from directly, each round uses one target sample and one sample from an easy-to-sample proposal distribution to compute the gains.
Importance weights correct for the difference between the proposal and the median measure, with the proposal chosen to keep these weights bounded.

The balance of the tests gives $\sum_iw_{t,i}g_{t,i}=0$, so the updates preserve the total weight without renormalization.
A standard multiplicative-weights argument~\cite{arora2012multiplicative}, combined with the Azuma--Hoeffding inequality, shows that each input distribution's average conditional gain is at most $O(\eps)$.
The tests also ensure that
\[
    D_{t,i}\leq\mathbb E[g_{t,i}\mid\cF_{t-1}]+\TV(P,H_i),
\]
where $\cF_{t-1}$ records the previous rounds.
Thus, with high probability, every distribution $H_{i}$ satisfies $\overline D_{i}\leq\TV(P,H_i)+O(\eps)$, where $\overline D_i=\frac1T\sum_t D_{t,i}$.
Applying the triangle inequality through $H_{i^*}$ then bounds the average distance between $U_t$ and $P$ by 
\[
    \frac1T\sum_{t=1}^T\frac12\|P-U_t\|_1
    \leq \overline D_{i^*}+\frac{1}{2}\|P-H_{i^*}\|_1
    \leq 2 \OPT+O(\eps).
\]

The same proposal samples provide bounded, unbiased estimates $\hat D_{t,i}$ of $D_{t,i}$, including a correction for the mass of $U_t$.
Azuma-Hoeffding inequality shows that, with high probability, for every input distribution, the average estimate is close to the average distance, namely $|\overline{\hat D}_{i}-\overline D_{i}|\leq O(\eps)$.
The algorithm returns the input distribution $H_{\hat i}$ with the smallest estimated average distance to the median measures.
Uniform concentration of the distance estimates and the choice of $\hat i$ imply
\[
    \overline D_{\hat i}
    \leq\overline{\hat D}_{\hat i}+O(\eps)
    \leq\overline{\hat D}_{i^*}+O(\eps)
    \leq\overline D_{i^*}+O(\eps)
    \leq\OPT+O(\eps).
\]
Averaging the triangle inequality through $U_t$ gives
\[
    \TV(P,H_{\hat i})
    \leq\frac1T\sum_{t=1}^T\frac12\|P-U_t\|_1+\overline D_{\hat i}
    \leq3\OPT+O(\eps).
\]
Tracking the constants yields $3\OPT+\eps$ with probability at least $1-\delta$ after $T=O(\eps^{-2}\log(n/\delta))$ rounds.
Each round evaluates the input densities and computes weighted medians only at the two sampled points, using $O(n)$ time and one target sample.

\section{Preliminaries}
\label{sec:preliminaries}

\subsection{Problem setting and computational model}
\label{sec:model}

Let $H_1,\ldots,H_n$ be probability distributions on a measurable space $(\mathcal X,\mathcal B)$, and let $P$ be the unknown target distribution on the same space.
Fix a common $\sigma$-finite dominating measure $\mu$ for $P,H_1,\ldots,H_n$, and write
\[
    p=\frac{dP}{d\mu},\qquad h_i=\frac{dH_i}{d\mu}\quad(i\in[n]).
\]
Here $h_i$ is the density of the distribution $H_i$ relative to $\mu$.
The target density $p$ is used only as mathematical notation and is not available to the algorithm.
We access $P$ through a sample oracle, and each $H_i$ through a sample oracle and a density oracle that evaluates $h_i(x)$ at any supplied point $x\in\mathcal X$.
Every sample-oracle call returns a fresh independent draw from the requested measure.

For accuracy $0<\eps\le1$ and failure probability $0<\delta<1$, our goal is to return a distribution $H_{\hat i}\in \{H_1,\ldots,H_n\}$ such that
\[
    \TV(P,H_{\hat i})\leq3\OPT+\eps,
    \qquad \OPT=\min_{i\in[n]}d_{\mathrm{TV}}(P,H_i),
\]
with probability at least $1-\delta$ over the samples and the algorithm's internal randomness.
The selection is \emph{proper} because the output must be one of the input distributions.

We use an exact-real oracle model: sample-oracle calls, density evaluations, arithmetic operations (addition, subtraction, multiplication, or division), comparisons, ceiling operations, and draws from uniform distributions each take unit time.
The algorithm's running time is the total number of such operations.
The sample complexity only counts the number of samples drawn from the target distribution $P$.

\subsection{Notation and density conventions}

Write $[n]=\{1,\ldots,n\}$, use $\ln$ for the natural logarithm and $\log$ for the base-two logarithm, and set $z_+=\max\{z,0\}$.
Throughout the paper, fix $i^*\in\argmin_{i\in[n]}d_{\mathrm{TV}}(P,H_i)$, breaking ties by a fixed rule.

Uppercase symbols denote measures and their lowercase counterparts denote densities relative to $\mu$.
In particular, $m$, $m'$, and $u$ denote the densities of the mixtures $M$, $M'$, and the median measure $U$ introduced in \cref{sec:algorithm}.

Since $h_i$ is the density of $H_i$ with respect to $\mu$, $H_i(E)=\int_E h_i\,d\mu$ denotes the probability of a measurable set $E\in\mathcal B$.
The same convention applies to $M$, $M'$, and $U$, although $U$ need not have mass one.

\subsection{Total variation and finite measures}
For a finite signed measure $V$ on $(\cX, \cB)$, write $\|V\|_1=|V|(\mathcal X)$ for its total variation norm.
If $V$ has density $v$ relative to $\mu$, then $\|V\|_1=\int|v|\,d\mu$.
For probability measures $A,B$, we use
\[
    d_{\mathrm{TV}}(A,B)=\sup_{E\in\mathcal B}|A(E)-B(E)|
        =\frac12\|A-B\|_1.
\]
For finite nonnegative measures $A,B$ and any measurable function $f$ with $|f|\leq1$, the total variation norm gives
\begin{equation}
\label{eq:bounded-test}
    \left|\int f\,d(A-B)\right|\leq\|A-B\|_1.
\end{equation}
We also use the triangle inequality for this norm, which applies even when the intermediate measure is not a probability measure.
In particular, for probability measures $A,B$ and a finite nonnegative measure $U$,
\[
    d_{\mathrm{TV}}(A,B)
        \leq\frac12\|A-U\|_1+\frac12\|U-B\|_1.
\]
When $U(\mathcal X)\neq1$, the quantity $\frac12\|U-A\|_1$ denotes half the total variation norm, rather than the total variation distance between two probability measures.

A useful identity accounts explicitly for the mass of such a measure.
If $U$ has density $u$ and mass $U(\mathcal X)=\int u\,d\mu$, then
\begin{align}
\label{eq:finite-measure-distance}
    \frac12\|U-H_i\|_1
        &=\frac12\int|u-h_i|\,d\mu \notag\\
        &=\int(u-h_i)_+\,d\mu-\frac12\int(u-h_i)\,d\mu \notag\\
        &=\int(u-h_i)_+\,d\mu+\frac{1-U(\mathcal X)}{2}.
\end{align}
Indeed, $|z|=2z_+-z$ and $\int h_i\,d\mu=1$ give this identity directly.
The correction $(1-U(\mathcal X))/2$ is needed because the median measures used by the algorithm are not normalized.

\section{Fast proper hypothesis selection}
In this section, we present a fast proper hypothesis selection algorithm, \cref{alg:proper}, and prove \cref{thm:proper}, which is restated below for convenience.
\cref{sec:overview} provides an overview of the algorithm.

\begin{theorem}[\cref{thm:proper} restated]
\label{thm:proper-restated}
Let $n\geq1$, $0<\eps\leq1$, and $0<\delta<1$. 
With probability at least $1-\delta$, \cref{alg:proper} returns a distribution $H_{\hat i}\in\cH$ which satisfies 
\[
    \TV(P,H_{\hat i})\leq3\OPT+\eps.
\]
The algorithm uses $O(\eps^{-2}\log(n/\delta))$ target samples, and runs in $O(n\eps^{-2}\log(n/\delta))$ time.
\end{theorem}

\subsection{Weighted medians and balanced median tests}
We now introduce the mixture of the input distributions defined by a weight vector, the median measure constructed from the mixture, and the balanced median tests that compare the median with each input distribution.
All of these objects are updated in each round of the algorithm, and the round subscript $t$ is omitted in this section.

Let $\bw=(w_1,\ldots,w_n)$ be a vector of positive weights with $\sum_i w_i=1$, and let $M=\sum_iw_iH_i$ be the corresponding mixture of input distributions.
Then $M$ is a probability measure with density $m=\sum_iw_ih_i$.

Let $u(x)$ be the \emph{lower weighted median} of $h_1(x),\ldots,h_n(x)$: the smallest value whose cumulative weight is at least $1/2$.
The function $u$ is measurable, since it is selected from finitely many measurable density functions using measurable comparisons.
Define the \emph{median measure} $U$ by $dU=u\,d\mu$.
Its mass $U(\mathcal X)$ need not be one, so we do not normalize it.

Let $W_<(x)=\sum_{i:h_i(x)<u(x)}w_i$ be the total weight on density values below $u(x)$, with $W_=(x)$ and $W_>(x)$ defined analogously. 
Define the \emph{balanced median tests} $\sigma_i(x)$ by
\begin{equation}
\label{eq:signed-tests}
    \sigma_i(x)=
    \begin{cases}
        1, & h_i(x)<u(x),\\
        \dfrac{W_>(x)-W_<(x)}{W_=(x)}, & h_i(x)=u(x),\\
        -1, & h_i(x)>u(x).
    \end{cases}
\end{equation}
The tests are measurable signed functions that indicate whether $h_i(x)$ is below or above the median, with a fractional tie value that exactly balances the total weight on each side.
Replacing the tie value by a fixed sign or by zero does not preserve $\sum_iw_i\sigma_i(x)=0$ in general, which is crucial for the multiplicative weight updates.
The denominator $W_=(x)$ is positive because the median is an input value and all weights are positive.

\begin{lemma}[Weighted-median identities]
\label{lem:weighted-median}
For every $x\in\mathcal X$ and $i\in[n]$,
\begin{align*}
    |\sigma_i(x)|&\leq1,\\
    \sum_iw_i\sigma_i(x)&=0,\\
    \sigma_i(x)(u(x)-h_i(x))&=|u(x)-h_i(x)|,\\
    0\leq u(x)&\leq2m(x).
\end{align*}
In particular, $U$ is a finite nonnegative measure with $U(\mathcal X)\leq2$.
\end{lemma}
\begin{proof}[Proof of \cref{lem:weighted-median}]
Fix $x\in\mathcal X$.
The lower-median definition gives $W_<(x)<1/2$ and $W_>(x)\leq1/2$.
Since $W_=(x)=1-W_<(x)-W_>(x)$, these inequalities imply
\[
    -W_=(x)\leq W_>(x)-W_<(x)\leq W_=(x).
\]
Thus the tie value lies in $[-1,1]$, proving the first identity.

Summing the three cases of \cref{eq:signed-tests} gives
\[
    \sum_iw_i\sigma_i(x)
    =W_<(x)-W_>(x)+W_=(x)\frac{W_>(x)-W_<(x)}{W_=(x)}=0.
\]
The third identity follows directly from the sign of $u(x)-h_i(x)$.

Finally, at least half the weight lies on values at least $u(x)$, so
\[
    m(x) 
    =\sum_{i\in [n]}w_ih_i(x)
    \geq u(x)\sum_{i:h_i(x)
    \geq u(x)}w_i
    \geq\frac{u(x)}2.
\]
Integrating gives $U(\mathcal X)\leq 2M(\cX)=2$ since $M$ is a probability measure.
\end{proof}

\subsection{Rounded weights and proposal distribution}
In each round, the algorithm compares the median measure with the target distribution to guide the weight updates. 
Sampling from the target is straightforward, but the median measure is not normalized and is difficult to sample from directly. 
We therefore introduce a proposal distribution which is easy to sample from and use importance weights to correct for the difference between the proposal and the median measure.
The subscript $t$ is also omitted in this section.

Directly sampling from the mixture $M$ is also difficult, since the mixture weights are arbitrary real numbers.
We thus round the weights to dyadic probabilities, which allows exact sampling from the rounded mixture $M'$ in linear time, using a uniform integer draw and a sample from one of the input distributions.

Choose the least power of two $N\geq n$, so $n\leq N<2n$, once for the entire algorithm.
Given $\bw$, first set $a_i=\lceil Nw_i\rceil$.
Increase $a_1$ by $2N-\sum_i a_i$, so that the final integer weights sum to $2N$, and define the rounded weights, mixture, and density by
\begin{equation*}
    w'_i=\frac{a_i}{2N},\qquad
    M'=\sum_iw'_iH_i,\qquad m'(x)=\sum_iw'_i h_i(x).
\end{equation*}

\begin{lemma}[Proposal properties]
\label{lem:proposal}
For every $x\in\mathcal X$ and $i\in[n]$,
\begin{equation*}
    w'_i\geq\frac{w_i}{2},\qquad
    w'_i\geq\frac1{2N},\qquad
    m'(x)\geq\frac{m(x)}{2},\qquad 
    m'(x)\geq \frac{u(x)}{4}.
\end{equation*}
\end{lemma}
\begin{proof}[Proof of \cref{lem:proposal}]
Before increasing $a_1$, positivity of the weights gives $\max\{Nw_i,1\}\leq a_i<Nw_i+1$, and hence $\sum_i a_i<N+n\leq2N$.
The remaining integer mass is nonnegative, and adding it to $a_1$ preserves $a_i\geq Nw_i$ and $a_i\geq1$.
Dividing by $2N$ gives the first two inequalities.

Multiplying $w'_i\geq w_i/2$ by $h_i(x)$ and summing gives $m'(x)\geq m(x)/2$. Together with $0\leq u(x)\leq2m(x)$ from \cref{lem:weighted-median}, we get $m'(x)\geq u(x)/4$.
\end{proof}

\paragraph{Sampling from $M'$.} 
Draw $I\sim\mathrm{Unif}([2N])$ in unit time.
Scan the cumulative sums $\sum_{k\leq i}a_k$ to find the unique index $i$ such that $\sum_{k<i}a_k<I\leq\sum_{k\leq i}a_k$, and return a fresh sample from $H_i$.
The index equals $i$ with probability $a_i/(2N)=w'_i$, so the procedure returns a sample from $M'$.
The sampling time is $O(n)$: the uniform draw and candidate-sample query each take unit time, and the cumulative scan takes $O(n)$ time.

\subsection{The iterative procedure}
\label{sec:iterative-procedure}
We now describe the iterative procedure, in which we use the balanced median tests to guide how to update the weights such that the new median measure is closer to the target.

Define the number of rounds
\begin{equation*}
    T=\left\lceil400\eps^{-2}\left\lceil\log\frac{3n}{\delta}\right\rceil\right\rceil.
\end{equation*}
Let $(\cF_t)_{t=0}^T$ be the filtration generated by the random choices in the algorithm, with $\cF_0$ containing the fixed input.
At round $t$, the weights $\bw_t=(w_{t,1},\ldots,w_{t,n})$, median measure $U_t$, tests $\sigma_{t,i}$, and proposal $M'_t$ are fixed.
Draw independent fresh samples $X_t\sim M'_t$ and $Y_t\sim P$, and define
\begin{align}
    \rho_t&=\frac{u_t(X_t)}{m'_t(X_t)},\label{eq:importance-weight}\\
    g_{t,i}&=\frac12\bigl(\rho_t\sigma_{t,i}(X_t)-\sigma_{t,i}(Y_t)\bigr),
        \label{eq:gain}\\
    \hat D_{t,i}&=\frac{(u_t(X_t)-h_i(X_t))_+}{m'_t(X_t)}
                    +\frac{1-\rho_t}{2}.\label{eq:distance}
\end{align}
Since every $w'_{t,i}>0$, if $m'_t(x)=0$ then all $h_i(x)$ and $u_t(x)$ vanish.
Ratios with denominator $m'_t(x)$ are defined to be zero.
This convention preserves the integral identities, and $M'_t$ assigns that set probability zero.
The expectation calculations in \cref{sec:conditional-expectations} make this explicit.

\paragraph{Gains and distance estimates.}
$\rho_t$ is the importance weight that corrects for the difference between the proposal and the median measure.
\cref{sec:conditional-expectations} proves that
\begin{equation*}
    \mathbb E[g_{t,i}\mid\mathcal F_{t-1}]
        =\frac12\int\sigma_{t,i}\,d(U_t-P),
    \qquad
    \mathbb E[\hat D_{t,i}\mid\mathcal F_{t-1}]
        =\frac12\|U_t-H_i\|_1.
\end{equation*}
By the definition of the tests, a positive gain $g_{t,i}$ suggests that $H_i$ lies, on average, on the same side of the median as the target, so increasing the weight of $H_i$ would move the median measure closer to the target, while a negative gain suggests the opposite.
Starting from uniform weights, we update the weights of the input distributions based on the gains.
It'll be shown that this update rule ensures that the average distance between the median measures and the target distribution across all rounds is small.

The distance estimates $\hat D_{t,i}$ are collected in each round to estimate the average distance between the median measures and each input distribution. 
Since $0\leq(u_t-h_i)_+\leq u_t\leq4m'_t$, it'll be shown that the distance estimate $\hat D_{t,i}$ lies in an interval of length four, independently of $n$ and the current weights.
This bound is crucial for the concentration analysis via the Azuma-Hoeffding inequality in \cref{lem:gain-distance-concentration}.
In contrast, the direct estimate $|u_t(X_t)-h_i(X_t)|/(2m'_t(X_t))$ introduces the ratio $h_i(X_t)/m'_t(X_t)$, whose bound $1/w'_{t,i}\leq2N$ depends on $n$.

\paragraph{Each round takes linear time.}
The algorithm represents $U_t$ implicitly through the weights and the densities of the input distributions.
In each round, it computes the weighted median and tests only at $X_t$ and $Y_t$.
Using the idea of median-of-medians~\cite{blum1973time}, \cref{sec:find-weighted-median} shows that the weighted median can be computed in $O(n)$ time.
Then the tests can be computed in $O(n)$ time by scanning the input densities and accumulating the weights.
Rounding, proposal sampling, density evaluation, and the updates also take $O(n)$ time, giving a total of $O(n)$ per round.

\begin{algorithm}[htbp]
\caption{Proper hypothesis selection via weighted medians}
\label{alg:proper}
\DontPrintSemicolon
\KwIn{Sample oracle for $P$; sample and density oracles for $H_1,\ldots,H_n$; accuracy $\eps$ and failure probability $\delta$.}
\KwOut{An input distribution $H_{\hat i}$.}
Choose $N$ as the least power of two at least $n$; set $T=\left\lceil400\eps^{-2}\left\lceil\log\frac{3n}{\delta}\right\rceil\right\rceil$\;
Initialize $w_{1,i}=1/n$ and $\overline{\hat D}_i=0$ for every $i\in[n]$\;
\For{$t=1,\ldots,T$}{
    Use $\bw_t$ to define $U_t$ and $\sigma_{t,i}$ by \cref{eq:signed-tests}\;
    Construct $\bw'_t$ and $M'_t$ using the rounding procedure of \cref{lem:proposal}\;
    Draw independent fresh samples $X_t\sim M'_t$ and $Y_t\sim P$\;
    Evaluate all $h_i(X_t),h_i(Y_t)$, the two weighted medians and their tests, and $m'_t(X_t)=\sum_iw'_{t,i}h_i(X_t)$\;
    Compute $\rho_t$ by \cref{eq:importance-weight}\;
    \For{$i=1,\ldots,n$}{
        Compute $g_{t,i}$ and $\hat D_{t,i}$ by \cref{eq:gain} and \cref{eq:distance}\;
        $w_{t+1,i}\gets w_{t,i}\bigl(1+(\eps/50)g_{t,i}\bigr)$\;
        $\overline{\hat D}_i\gets \overline{\hat D}_i+\hat D_{t,i}/T$\;
    }
}
Return $H_{\hat i}$ for $\hat i\in\argmin_{i\in[n]}\overline{\hat D}_i$\;
\end{algorithm}

\subsection{Analysis of the algorithm}
In this section, we analyze the performance of \cref{alg:proper} and prove \cref{thm:proper-restated}.
Each round of the algorithm draws one sample from the target distribution, and runs in $O(n)$ time.
The total sample complexity is $T=O(\eps^{-2}\log(n/\delta))$, and the total time complexity is $O(nT)=O(n\eps^{-2}\log(n/\delta))$, as claimed in \cref{thm:proper-restated}.
We now show that the output distribution $H_{\hat i}$ satisfies $\TV(P,H_{\hat i})\leq3\OPT+\eps$ with probability at least $1-\delta$.

\subsubsection{Bounded estimates and validity of the updates}
Suppose that the current weights are positive and sum to one.
By \cref{lem:proposal}, $m'_t\geq u_t/4$, so $0\leq\rho_t\leq4$ and
\[
    0\leq\frac{(u_t(X_t)-h_i(X_t))_+}{m'_t(X_t)}\leq\rho_t.
\]
Together with $|\sigma_{t,i}|\leq1$ from \cref{lem:weighted-median}, these bounds imply
\begin{equation}
\label{eq:estimator-ranges}
    |g_{t,i}|\leq\frac52,\qquad
    -\frac32\leq\hat D_{t,i}\leq\frac52.
\end{equation}

\paragraph{Validity of the update.}
The initial weights are positive and sum to one.
The balance identity $\sum_iw_i\sigma_i(x)=0,$ in \cref{lem:weighted-median} holds pointwise, and therefore
\[
    \sum_iw_{t,i}g_{t,i}
    =\frac12\left(\rho_t\sum_iw_{t,i}\sigma_{t,i}(X_t)
                     -\sum_iw_{t,i}\sigma_{t,i}(Y_t)\right)=0.
\]
Consequently,
\[
    \sum_iw_{t+1,i}
    =\sum_iw_{t,i}+\frac{\eps}{50}\sum_iw_{t,i}g_{t,i}=1.
\]
Moreover, $|(\eps/50)g_{t,i}|\leq\eps/20\leq1/20$, so every update factor is positive.
Induction shows that all weight vectors, medians, and proposals are well defined, and that \cref{eq:estimator-ranges} holds in every round.

With this in place, the following standard multiplicative-weights argument~\cite{arora2012multiplicative} bounds the cumulative gain of each input distribution.

\begin{lemma}[Deterministic gain bound]
\label{lem:gain-bound}
For every $i\in[n]$,
\begin{equation*}
    \frac1T\sum_{t=1}^Tg_{t,i}
    \leq\frac{50\ln n}{\eps T}+\frac{\eps}{8}
    \leq\frac{7\eps}{32}.
\end{equation*}
\end{lemma}
\begin{proof}[Proof of \cref{lem:gain-bound}]
Since $w_{t+1,i}=w_{t,i}(1+(\eps/50)g_{t,i})$ and $w_{1,i}=1/n$, it holds that
\begin{align*}
    \ln w_{T+1,i}
    &=-\ln n+\sum_{t=1}^T\ln\bigl(1+(\eps/50)g_{t,i}\bigr)\\
    &\geq-\ln n+\frac{\eps}{50}\sum_{t=1}^Tg_{t,i}
                    -\frac{\eps^2}{2500}\sum_{t=1}^Tg_{t,i}^2\\
    &\geq-\ln n+\frac{\eps}{50}\sum_{t=1}^Tg_{t,i}-\frac{\eps^2T}{400},
\end{align*}
where the first inequality uses $\ln(1+x)\geq x-x^2$ for $|x|\leq1/2$, and the second uses $g_{t,i}^2\leq25/4$.
The weights are positive and sum to one, so $0<w_{T+1,i}\leq1$ and $\ln w_{T+1,i}\leq0$.
Rearranging gives
\[
    \frac1T\sum_{t=1}^Tg_{t,i}
    \leq\frac{50\ln n}{\eps T}+\frac{\eps}{8}.
\]
Finally, since $\ln 2<3/4$, we have $\ln n\leq\ln(3n/\delta)\leq3\lceil\log(3n/\delta)\rceil/4$.
Together with
$T\geq400\eps^{-2}\lceil\log(3n/\delta)\rceil$, this gives
\[
    \frac{50\ln n}{\eps T}+\frac{\eps}{8}
        \leq\frac{3\eps}{32}+\frac{\eps}{8}
        =\frac{7\eps}{32}.
\]
\end{proof}

\subsubsection{Conditional expectations}
\label{sec:conditional-expectations}
Recall the filtration $\mathcal F_t$ from \cref{sec:iterative-procedure}.
Define
\[
    D_{t,i}=\frac12\|U_t-H_i\|_1,
    \qquad \overline D_i=\frac1T\sum_{t=1}^T D_{t,i}.
\]
Note that $U_t$ need not have mass one.

\begin{lemma}[Conditional expectations]
\label{lem:distance-estimate-expectations}
For every round $t\in[T]$ and coordinate $i\in[n]$,
\begin{align*}
    D_{t,i}\leq\mathbb E[g_{t,i}\mid\mathcal F_{t-1}]+d_{\mathrm{TV}}(P,H_i),\qquad
    \mathbb E[\hat D_{t,i}\mid\mathcal F_{t-1}]=D_{t,i}.
\end{align*}
\end{lemma}
\begin{proof}[Proof of \cref{lem:distance-estimate-expectations}]
Conditional on $\mathcal F_{t-1}$, the current weights, median, tests, and proposal are fixed, and $X_t\sim M'_t$ and $Y_t\sim P$ are fresh and independent.
By the sign identity $\sigma_i(x)(u_t(x)-h_i(x))=|u_t(x)-h_i(x)|$ in \cref{lem:weighted-median},
\begin{align}
    D_{t,i}
    &=\frac12\int|u_t(x)-h_i(x)|\,d\mu(x)\notag\\
    &=\frac12\int\sigma_{t,i}\,d(U_t-H_i)\notag\\
    &=\frac12\int\sigma_{t,i}\,d(U_t-P)+\frac12\int\sigma_{t,i}\,d(P-H_i).
    \label{eq:gain-integral}
\end{align}
Importance sampling gives
\begin{align*}
    \mathbb E[\rho_t\sigma_{t,i}(X_t)\mid\mathcal F_{t-1}]
        =\int_{\{x:m'_t(x)>0\}}
            \frac{u_t(x)\sigma_{t,i}(x)}{m'_t(x)}m'_t(x)\,d\mu(x)
        =\int\sigma_{t,i}\,dU_t.
\end{align*}
On the set $\{x:m'_t(x)=0\}$, which has measure zero under $M'_t$, we have already defined the ratio $\rho_t=u_t(x)/m'_t(x)$ to be zero. The expectation can therefore be written as an integral over ${x:m'_t(x)>0}$. After canceling $m'_t(x)$ on this set, we obtain the last equality by extending the integral to all of $\mathcal X$, since $u_t(x)=0$ wherever $m'_t(x)=0$.
Since $\mathbb E[\sigma_{t,i}(Y_t)\mid\mathcal F_{t-1}]=\int\sigma_{t,i}\,dP$, we obtain
\begin{equation}
\label{eq:gain-expectation}
    \mathbb E[g_{t,i}\mid\mathcal F_{t-1}]
        =\frac12\int\sigma_{t,i}\,d(U_t-P).
\end{equation}
Combining \cref{eq:gain-integral,eq:gain-expectation} and applying \cref{eq:bounded-test} yields
\begin{align*}
    D_{t,i}
    &=\mathbb E[g_{t,i}\mid\mathcal F_{t-1}]
        +\frac12\int\sigma_{t,i}\,d(P-H_i)\\
    &\leq\mathbb E[g_{t,i}\mid\mathcal F_{t-1}]+d_{\mathrm{TV}}(P,H_i).
\end{align*}

For the distance estimate, \cref{eq:finite-measure-distance} gives
\begin{equation*}
    D_{t,i}=\int(u_t(x)-h_i(x))_+\,d\mu(x)
        +\frac{1-U_t(\mathcal X)}{2}.
\end{equation*}
Both $u_t$ and $(u_t-h_i)_+$ vanish wherever $m'_t=0$, so the same importance-sampling calculation yields
\begin{align*}
    \mathbb E\!\left[\frac{(u_t(X_t)-h_i(X_t))_+}{m'_t(X_t)}
        \,\middle|\,\mathcal F_{t-1}\right]
        &=\int(u_t(x)-h_i(x))_+\,d\mu(x),\\
    \mathbb E[\rho_t\mid\mathcal F_{t-1}]&=U_t(\mathcal X).
\end{align*}
Substituting into \cref{eq:distance} proves
\begin{equation*}
    \mathbb E[\hat D_{t,i}\mid\mathcal F_{t-1}]=D_{t,i}.
\end{equation*}
\end{proof}

\subsubsection{Concentration bounds}
We now show that the average gains and distance estimates concentrate around their conditional expectations.
The weights, weighted medians and proposals depend on previous samples, so we use a conditional form of the Azuma--Hoeffding inequality~\cite{azuma1967weighted,hoeffding1963probability}.

\begin{lemma}[Azuma--Hoeffding inequality]
\label{lem:concentration}
Let $(\cF_t)_{t=0}^T$ be a filtration and let $(\xi_t)_{t=1}^T$ be a sequence of random variables, each $\mathcal F_t$-measurable with $\mathbb E[\xi_t\mid\mathcal F_{t-1}]=0$.
Suppose that $a_t\leq\xi_t\leq b_t$ almost surely, where $a_t,b_t$ are $\mathcal F_{t-1}$-measurable and $b_t-a_t\leq c$ for a constant $c>0$.
Then for $s>0$,
\begin{equation*}
    \Pr\!\left\{\frac1T\sum_{t=1}^T\xi_t>s\right\}
    \leq\exp\!\left(-\frac{2Ts^2}{c^2}\right).
\end{equation*}
The same bound holds for the lower tail.
\end{lemma}

\begin{lemma}[Concentration of gains and distance estimates]
\label{lem:gain-distance-concentration}
With probability at least $1-\delta$, the following inequalities hold simultaneously for every $i\in[n]$:
\begin{align*}
    \frac1T\sum_{t=1}^T\mathbb E[g_{t,i}\mid\mathcal F_{t-1}]
    &\leq\frac1T\sum_{t=1}^Tg_{t,i}+\frac{5\eps}{32},\\
    |\overline{\hat D}_i-\overline D_i|
    &\leq\frac{\eps}{8}.
\end{align*}
\end{lemma}
\begin{proof}[Proof of \cref{lem:gain-distance-concentration}]
$T=\left\lceil400\eps^{-2}\left\lceil\log\frac{3n}{\delta}\right\rceil\right\rceil$ and $\ln 2<3/4$ imply
$\frac{T\eps^2}{512}
    \geq\frac{25}{32}\left\lceil\log\frac{3n}{\delta}\right\rceil
    \geq\ln\frac{3n}{\delta}.$

For a fixed $i$, the differences $\xi_t=\mathbb E[g_{t,i}\mid\mathcal F_{t-1}]-g_{t,i}$ have conditional mean zero and conditional range of length five by \cref{eq:estimator-ranges}.
Apply \cref{lem:concentration} with $c=5$ and $s=5\eps/32$ to bound the probability that the first inequality fails for this $i$ by $\exp(-T\eps^2/512)\leq\delta/(3n)$.
A union bound over $i\in[n]$ gives failure probability at most $\delta/3$.

By \cref{lem:distance-estimate-expectations}, $\hat D_{t,i}-D_{t,i}$ has conditional mean zero, and its conditional range has length four by \cref{eq:estimator-ranges}.
Apply both tails of \cref{lem:concentration} with $c=4$ and $s=\eps/8$ to bound the probability that the second inequality fails for this $i$ by $2\exp(-T\eps^2/512)\leq2\delta/(3n)$.
A union bound over $i\in[n]$ gives failure probability at most $2\delta/3$.

Combining the two failure probabilities proves the simultaneous guarantee.
\end{proof}

Combining \cref{lem:distance-estimate-expectations,lem:gain-bound,lem:gain-distance-concentration} gives the following corollary.

\begin{corollary}
    \label{cor:average-distance}
    With probability at least $1-\delta$, the following inequalities hold simultaneously for every $i\in[n]$:
    \begin{equation*}
        \overline D_i\leq d_{\mathrm{TV}}(P,H_i)+\frac{3\eps}{8},\qquad
        |\overline{\hat D}_i-\overline D_i|\leq\frac{\eps}{8}.
    \end{equation*}
\end{corollary}
\begin{proof}[Proof of \cref{cor:average-distance}]
On the event of \cref{lem:gain-distance-concentration}, for every $i\in[n]$, we already have the second inequality $|\overline{\hat D}_i-\overline D_i|\leq\frac{\eps}{8}$. 
The first inequality follows from the following chain of inequalities:
\begin{align*}
    &\overline D_i
    \\ 
    (\cref{lem:gain-distance-concentration})&\leq \overline{\hat D}_i + \frac{\eps}{8}
    \\
    (\cref{lem:distance-estimate-expectations})&\leq \frac1T\sum_{t=1}^T\mathbb E[g_{t,i}\mid\mathcal F_{t-1}]
        + d_{\mathrm{TV}}(P,H_i) + \frac{\eps}{8}
    \\
    (\cref{lem:gain-distance-concentration})&\leq \frac1T\sum_{t=1}^Tg_{t,i}+\frac{5\eps}{32}+ d_{\mathrm{TV}}(P,H_i) + \frac{\eps}{8}
    \\
    (\cref{lem:gain-bound})&\leq d_{\mathrm{TV}}(P,H_i)+\frac{3\eps}{8}.
\end{align*}
\end{proof}

\subsubsection{Approximation guarantee}
On the event in \cref{cor:average-distance}, which has probability at least $1-\delta$.
Fix $i^*\in\argmin_i d_{\mathrm{TV}}(P,H_i)$.
Since the output index $\hat i$ minimizes the estimated average distance, the average distance between the median measures and the output distribution is bounded by
\begin{align*}
    \overline D_{\hat i}
    \leq\overline{\hat D}_{\hat i}+\frac{\eps}{8}
    \leq\overline{\hat D}_{i^*}+\frac{\eps}{8}
    \leq\overline D_{i^*}+\frac{\eps}{4}
    \leq\mathrm{OPT}+\frac{5\eps}{8}.
\end{align*}
where the first, third, and fourth inequalities follow from \cref{cor:average-distance}, and the second inequality follows from the definition of $\hat i$.
Moreover, the median measures are close to the best distribution on average, and thus close to the target distribution, since
\begin{align*}
    \frac1T\sum_{t=1}^T\frac12\|P-U_t\|_1
    &\leq \frac1T\sum_{t=1}^T \frac12\left(\|P-H_{i^*}\|_1+\|H_{i^*}-U_t\|_1\right)\\
    &\leq d_{\mathrm{TV}}(P,H_{i^*})+\overline D_{i^*}\\
    &\leq2\mathrm{OPT}+\frac{3\eps}{8}.
\end{align*}
where the last inequality uses \cref{cor:average-distance}.
For every $t$, use $U_t$ as the intermediate measure in the triangle inequality:
\[
    d_{\mathrm{TV}}(P,H_{\hat i})
    \leq\frac12\|P-U_t\|_1+\frac12\|U_t-H_{\hat i}\|_1
    =\frac12\|P-U_t\|_1+D_{t,\hat i}.
\]
Averaging the triangle inequality through $T$ rounds yields
\begin{align*}
    d_{\mathrm{TV}}(P,H_{\hat i})
    &\leq\frac1T\sum_{t=1}^T\frac12\|P-U_t\|_1+\overline D_{\hat i}\\
    &\leq2\mathrm{OPT}+\frac{3\eps}{8}+\mathrm{OPT}+\frac{5\eps}{8}\\
    &=3\mathrm{OPT}+\eps.
\end{align*}

\section*{AI Disclosure}
ChatGPT was used extensively during the preparation of this work. 
The author suggested using a pointwise median to obtain an improper distribution with approximation factor $2$ on the lower-bound instance of~\cite{aamand2025fast}. 
Then, ChatGPT proposed combining weighted medians with multiplicative weight updates. 
ChatGPT and Claude also assisted with exposition and editing. 
The author takes full responsibility for the results, proofs, and references.

% \section*{Acknowledgements}

% \newpage
\appendix
\section{Finding the weighted median in linear time}
\label{sec:find-weighted-median}

The median-of-medians algorithm of Blum et al.~\cite{blum1973time} finds the $k$-th smallest of $n$ values in deterministic $O(n)$ time, without sorting all the values.
We can adapt this algorithm to find the weighted median of $n$ values with positive weights in $O(n)$ time, without sorting all the values, as shown below. 

\begin{lemma}
\label{lem:find-weighted-median}
Given real values $v_1,\ldots,v_n$ with positive weights $w_1,\ldots,w_n$, satisfying $\sum_{i=1}^n w_i=1$, the smallest input value whose cumulative weight is at least $1/2$ can be found deterministically in $O(n)$ worst-case time.
\end{lemma}

\begin{proof}[Proof of \cref{lem:find-weighted-median}]
Initialize an active set $A=\{1,\ldots,n\}$ and a residual target $\tau=1/2$.
Use the median-of-medians algorithm to select the ordinary median of the active values, i.e., find the $\lceil |A|/2\rceil$-th smallest value.
Write $z$ for its value, scan the active set and partition $A$ into
\[
  A_< =\{i\in A:v_i<z\},\qquad
  A_= =\{i\in A:v_i=z\},\qquad
  A_> =\{i\in A:v_i>z\},
\]
and compute $W_L=\sum_{i\in A_<}w_i$ and $W_E=\sum_{i\in A_=}w_i$.
Apply the following rules:
\begin{enumerate}
  \item If $\tau\leq W_L$, retain $A_<$ and leave $\tau$ unchanged.
  \item If $W_L<\tau\leq W_L+W_E$, return $z$.
  \item If $\tau>W_L+W_E$, retain $A_>$ and replace $\tau$ by $\tau-W_L-W_E$.
\end{enumerate}
These rules ensure that the active set always contains the smallest value whose cumulative weight is at least $1/2$.
Since $z$ represents the median value of the active set, each nonterminal step at least halves the active-set size, including in the presence of ties.
Thus the procedure terminates after at most $\lceil\log n\rceil$ steps.
Finding the median, computing the cumulative weights, and partitioning take $O(|A|)$ operations per step, for a total of
\[
  O(n)+O(n/2)+O(n/4)+\cdots=O(n),
\]
which is the desired worst-case bound.
\end{proof}

% \newpage
\addcontentsline{toc}{section}{References}
\bibliographystyle{alpha}
\bibliography{references}

\end{document}